\documentclass[11pt, twoside, letterpaper]{amsart}
\usepackage{amsmath, hyperref}

\newtheorem{thm}{Theorem}[section]

\newtheorem{lem}[thm]{Lemma}

\theoremstyle{definition}
\newtheorem{defn}[thm]{Definition}

\theoremstyle{remark}
\newtheorem{rem}[thm]{Remark}

\numberwithin{equation}{section}

\newcommand{\set}[1]{\left\{#1\right\}}
\newcommand{\Real}{\mathbb R}
\newcommand{\Natural}{\mathbb N}

\newcommand{\such}{\ | \ }
\newcommand{\nin}{n \in \Natural}
\newcommand{\prob}{\mathbb{P}}
\newcommand{\Exp}{\mathcal E}

\newcommand{\qprobu}{{\mathbb{P}_u}}

\newcommand{\expec}{\mathbb{E}}

\newcommand{\expecqu}{\expec_u}

\newcommand{\basispwf}{(\Omega, \, \F, \, \bF, \, \prob)}
\newcommand{\basis}{(\Omega, \, \bF)}
\newcommand{\basisg}{(\Omega, \, \bG)}

\newcommand{\basisp}{(\Omega, \, \bF, \, \prob)}

\newcommand{\basisgp}{(\Omega, \, \bG, \, \prob)}

\newcommand{\basisqu}{(\Omega, \, \bF, \, \qprobu)}

\newcommand{\F}{\mathcal{F}}
\newcommand{\G}{\mathcal{G}}

\newcommand{\ud}{\mathrm d}

\newcommand{\limt}{\lim_{t \to \infty}}

\newcommand{\zi}{{\Real_+}}

\newcommand{\zd}{[0, \cdot]}

\newcommand{\zo}{[0, 1)}

\newcommand{\trace}{\mathsf{trace}}

\newcommand{\cS}{{}^\mathsf{c} \kern-0.21em S}
\newcommand{\cH}{{}^\mathsf{c} \kern-0.23em H}
\newcommand{\cMM}{{}^\mathsf{c} \kern-0.19em [M, M]}

\newcommand{\pare}[1]{\left(#1\right)}
\newcommand{\bra}[1]{\left[#1\right]}
\newcommand{\dbra}[1]{[\kern-0.15em[ #1 ]\kern-0.15em]}
\newcommand{\dbraco}[1]{[\kern-0.15em[ #1 [\kern-0.15em[}
\newcommand{\dbraoc}[1]{]\kern-0.15em] #1 ]\kern-0.15em]}
\newcommand{\dbraoo}[1]{]\kern-0.15em] #1 [\kern-0.15em[}

\newcommand{\hX}{\widehat{X}}

\newcommand{\hY}{\widehat{Y}}

\newcommand{\X}{\mathcal{X}}

\newcommand{\bF}{\mathbf{F}}

\newcommand{\dfn}{\, := \,}

\newcommand{\bG}{\mathbf{G}}
\newcommand{\indic}{\mathbb{I}}

\newcommand{\num}{num\'eraire}
\begin{document}

\title{A time before which insiders would not undertake risk}%
\author{Constantinos Kardaras}%
\address{Constantinos Kardaras, Mathematics and Statistics Department, Boston University, 111 Cummington Street, Boston, MA 02215, USA.}%
\email{kardaras@bu.edu}%

\thanks{The author would like to thank an anonymous referee for constructive remarks that helped improve the paper. This work is supported in part by the National Science Foundation under grant number DMS-0908461.}

\date{\today}%

\begin{abstract}
A continuous-path semimartingale market model with wealth processes discounted by a riskless asset is considered. The \num \ portfolio is the unique strictly positive wealth process that, when used as a benchmark to denominate all other wealth, makes all wealth processes local martingales. It is assumed that the \num \ portfolio exists and that its wealth increases to infinity as time goes to infinity. Under this setting, an initial enlargement of the filtration is performed, by including the overall minimum of the \num \ portfolio. It is established that all nonnegative wealth processes, when stopped at the time of the overall minimum of the \num \ portfolio, become local martingales in the enlarged filtration. This implies that risk-averse insider traders would refrain from investing in the risky assets before that time. A partial converse to the previous result is also established in the case of complete markets, showing that the time of the overall minimum of the \num \ portfolio is in a certain sense unique in rendering undesirable the act of undertaking risky positions before it. The aforementioned results shed light to the importance of the \num \ portfolio as an indicator of overall market performance.
\end{abstract}

\maketitle

\setcounter{section}{-1}

\section{Introduction}

When modeling insider trading, one usually enlarges the ``public'' information flow by including knowledge of a non-trivial random variable, which represents the extra information of the insider, from the very beginning. (This method called \emph{initial filtration enlargement}, as opposed to \emph{progressive filtration enlargement} --- for more details, see \cite[Chapter VI]{MR1037262}.) It is then of interest to explore the effect that the extra information has on the trading behavior of the insider --- for an example, see \cite{MR2223957}. Under this light, the topic of the present paper may be considered slightly unorthodox, as we identify an initial filtration enlargement and a stopping time of the enlarged filtration (which is \emph{not} a stopping time of the original filtration) with the property that risk-averse insider traders would refrain from taking risky positions before that time. As will be revealed, this apparently ``negative'' result, though not helpful in the theory of insider trading, sheds more light to the importance of a specific investment opportunity, namely, the \num \ portfolio.

Our setting is a continuous-path semimartingale market model with $d$ asset-price processes $S^1, \ldots, S^d$. All wealth is discounted with respect to some locally riskless asset. Natural structural assumptions are imposed --- in particular, we only enforce a mild market viability condition, and allow for the existence of some discounted wealth process that will grow unconditionally as time goes to infinity. Such assumptions are satisfied in every reasonable infinite time-horizon model. In such an environment, the \num \ portfolio --- an appellation coined in \cite{LONG} --- is the unique nonnegative wealth process $\hX$ with unit initial capital such that all processes $S^i / \hX$, $i \in \set{1, \ldots, d}$, become local martingales. The \num \ portfolio has several interesting optimality properties. For instance, it maximizes expected logarithmic utility for all time-horizons and achieves maximal long-term growth --- for more information, check \cite{MR2335830}. The goal of the present paper is to add yet one more to the remarkable list of properties of the \num \ portfolio.

The original filtration  $\bF$ is  enlarged to $\bG$, which further contains information on the overall minimum level $\min_{t \in \Real_+} \hX(t)$ of the \num \ portfolio. In particular, the time $\rho$ that this overall minimum is achieved (which can be shown to be almost surely unique) becomes a stopping time with respect to $\bG$. Our first main result states that all $S^i$, $i \in \set{1, \ldots, d}$, become local martingales up to time $\rho$ under the \emph{enlarged filtration} $\bG$ and \emph{original probability} $\prob$. Note that the asset-price processes are discounted by the locally riskless wealth process, and not by the \num \ portfolio. (The latter discounting makes asset price-processes local martingales under $(\bF, \prob)$, while the former discounting makes asset price-processes, when stopped at $\rho$, local martingales under $(\bG, \prob)$.) In essence, $\prob$ becomes a risk-neutral measure for the model with enlarged filtration up to time $\rho$. An immediate consequence of this fact is that a risk-averse investor would refrain from taking risky positions up to time $\rho$, since they would result in no compensation for the risk that is being undertaken, in terms of excess return relative to the riskless account. (Note, however, that an insider can arbitrage unconditionally after time $\rho$ with \emph{no} downside risk whatsoever involved, simply by taking arbitrarily large long positions in the the \num \ portfolio immediately after $\rho$.) In effect, trading in the market occurs simply because traders do not have information about the time of the overall minimum of the \num \ portfolio. In fact, until time $\rho$, not only the \num \ portfolio, but the whole market performs badly, since the expected outcome of any portfolio at time $\rho$ is necessarily less or equal than the initial capital used to set it up.

A partial converse to the previous result is also presented. Under an extra completeness assumption on the market, it is shown that if a random time $\phi$ (satisfying a couple of technical properties) is such that $\expec [X(\phi)] \leq X(0)$ holds for any nonnegative wealth process $X$ formed by trading with information $\bF$, then $\phi$ is necessarily equal to the time of the overall minimum of the \num \ portfolio. Combined with our first main result, this clarifies the unique role of the \num \ portfolio as an indicator of overall market performance.

The structure of the remainder of the paper is simple. In Section \ref{sec: results} the results are presented, while Section \ref{sec: proof} contains the proofs.

\section{Results} \label{sec: results}

\subsection{The set-up} \label{subsec: set-up}

Let $\basispwf$ be a filtered probability space --- here, $(\Omega, \, \F, \, \prob)$ is a complete probability space and $\bF = (\F(t))_{t \in \Real_+}$ is a right-continuous filtration such that, for each $t \in \Real_+$, $\F(t) \subseteq \F$ and $\F(t)$ contains all  $\prob$-null sets of $\F$ --- in other words, $\bF$ satisfies the \emph{usual conditions}. Without affecting in any way the generality of our discussion, we shall be assuming that $\F(0)$ is trivial modulo $\prob$. Relationships involving random variables are to be understood in the $\prob$-a.s. sense; relationships involving processes hold modulo evanescence.

On $\basisp$, let $S = (S^i)_{i=1, \ldots, d}$ be a vector-valued semimartingale with continuous paths. For each $i \in \set{1, \ldots, d}$, $S^i$ represents the discounted, with respect to some baseline security, price of a liquid asset in the market. The baseline security, which we shall simply call \textsl{discounting process}, should be thought as a locally riskless account. In contrast, the other assets are supposed to represent riskier investments. We also set $S^0 \dfn 1$ to denote the wealth accumulated by the baseline locally riskless security, discounted by itself.

Starting with capital $x \in \Real_+$, and investing according to some $d$-dimensional, $\bF$-predictable and $S$-integrable strategy $\vartheta$ modeling the number of liquid assets held in the portfolio, an economic agent's discounted wealth is given by $X^{x, \vartheta} = x + \int_0^\cdot \vartheta^\top(t) \ud S(t)$. Define $\X_\bF (x)$ as the set of all processes  $X^{x, \vartheta}$ in the previous notation that remain nonnegative at all times. Furthermore, we set $\X_\bF \dfn \bigcup_{x \in \Real_+} \X_\bF(x)$. 

Below, we gather some definitions and results that have appeared previously in the literature. More information about them can be found in \cite{MR2335830} and, for the special case of continuous-path semimartingales that is considered here, in \cite[Section 4]{Kar_09_fin_add_ftap}.

\begin{defn} \label{dfn: arb first kind}
We shall say that the market allows for \textsl{arbitrage of the first kind} if there exists $T \in \Real_+$ and an $\F(T)$-measurable random variable $\xi$ with $\prob[\xi \geq 0] = 1$, $\prob[\xi > 0] > 0$, such that for all $x > 0$ one can find $X \in \X(x)$ satisfying $\prob[X(T) \geq \xi] = 1$. If the market does not allow for any arbitrage of the first kind, we say that condition NA$_1$ holds.
\end{defn}

Condition NA$_1$ is weaker than the ``No Free Lunch with Vanishing Risk'' market viability condition of \cite{MR1304434}, and is actually equivalent to the requirement that $\lim_{\ell \to \infty} \sup_{X \in \X^\bF (x)} \prob \bra{X(T) > \ell} = 0$ holds for all $x \in \Real_+$ and $T \in \Real_+$ --- see \cite[Proposition 1]{Kar_09_fin_add_ftap}. The latter boundedness-in-probability requirement is coined condition BK in \cite{MR1647282} and condition ``No Unbounded Profit with Bounded Risk'' (NUPBR) in \cite{MR2335830}.

\begin{defn} \label{dfn: smart defl, num}
A strictly positive local martingale deflator is a strictly positive process $Y$ with $Y(0) = 1$ such that $Y S^i$ is a local martingale on $\basisp$ for all $i \in \set{0, \ldots , d}$. (The last requirement is equivalent to asking that $Y X$ is a local martingale on $\basisp$ for all $X \in \X_\bF$.) A strictly positive process $\hX \in \X_\bF(1)$ will be called \textsl{\textsl{the \num} portfolio} if $\hY \dfn 1/ \hX$ is a (necessarily, strictly positive) local martingale deflator.
\end{defn}

By Jensen's inequality, it is straightforward to see that if the \num \ portfolio $\hX$ exists, then it is unique. Obviously, if the \num \ portfolio exists then at least one strictly positive local martingale deflator exists in the market. Interestingly, the converse also holds, i.e., existence of the \num \ portfolio is equivalent to existence of at least one strictly positive local martingale deflator. Furthermore, the previous are also equivalent to condition NA$_1$ holding in the market.

Condition NA$_1$ can also be described in terms of the asset-prices process drifts and volatilities. More precisely, let $A = (A^1, \ldots, A^d)$ be the continuous-path finite-variation process appearing in the Doob-Meyer decomposition of the continuous-path semimartingale $S$. For $i \in \set{1, \ldots, d}$ and $k \in \set{1, \ldots, d}$, denote by $[S^i, S^k]$ the quadratic (co)variation of $S^i$ and $S^k$. Also, let $[S, S]$ be the $d \times d$ nonnegative-definite symmetric matrix-valued process whose $(i, k)$-component is $[S^i, S^k]$ for $i \in \set{1, \ldots, d}$ and $k \in \set{1, \ldots, d}$. Call now $G := \trace [S, S]$, where $\trace$ is the operator returning the trace of a matrix. Observe that $G$ is an increasing, adapted, continuous process, and that there exists a $d \times d$ nonnegative-definite symmetric matrix-valued process $c$ such that $[S^i, S^k] = \int_0^\cdot c^{i,k}(t) \ud G(t)$ for $i \in \set{1, \ldots, d}$ and $k \in \set{1, \ldots, d}$; $[S, S] = \int_0^\cdot c(t) \ud G(t)$ in short. Then, condition NA$_1$ is equivalent to the existence of a $d$-dimensional, predictable process $\xi$ such that $A = \int_0^\cdot (c(t) \xi(t)) \ud G(t)$, satisfying $\int_0^T \pare{\xi^\top(t) c(t) \xi(t)} \ud G(t) < \infty$ for all $T \in \Real_+$. In fact, with the previous notation, it can be checked that the \num \ portfolio is given by $\hX = \Exp \pare{\int_0^\cdot \xi^\top(t) \ud S(t)}$, where ``$\Exp$'' denotes the stochastic exponential operator.

\begin{defn} \label{dfn: disc process}
We shall say that \textsl{the discounting process is asymptotically suboptimal} if there exists $X \in \X_\bF$ such that $\prob \bra{\lim_{t \to \infty} X(t) = \infty}  = 1$.
\end{defn}

The previous definition is self-explanatory --- the locally riskless discounting process (which is used as a baseline to denominate all other wealth) is asymptotically suboptimal if it can be beaten unconditionally in the long run by some other wealth process in the market. As a simple example where the discounting process is asymptotically suboptimal, we mention any multi-dimensional Black-Scholes model such that the probability $\prob$ is not a risk-neutral one.

Given the existence of the \num \ portfolio $\hX$ (i.e., under the validity of condition NA$_1$), the discounting process is asymptotically suboptimal if and only if $\prob \big[ \lim_{t \to \infty} \hX(t) = \infty \big]  = 1$. Indeed, if there exists some $X \in \X_\bF$ such that $\prob \bra{\lim_{t \to \infty} X(t) = \infty}  = 1$, the supermartingale property of $X / \hX$ and Doob's nonnegative supermartingale convergence theorem give $\prob \big[ \lim_{t \to \infty} \hX(t) = \infty \big]  = 1$. Furthermore,  under condition NA$_1$, and with the notation used in the paragraph right before Definition \ref{dfn: disc process}, it can be checked that the discounting process is asymptotically suboptimal if and only if $\int_0^\infty \pare{\xi^\top(t) c(t) \xi(t)} \ud G(t) = \infty$.

\subsection{The first result} \label{subsec: main result}

For the purposes of \S \ref{subsec: main result}, assume that condition NA$_1$ holds in the market and the the discounting process is asymptotically suboptimal. Recall that this is equivalent to existence of the \num \ portfolio $\hX$, which satisfies $\prob \big[ \limt \hX(t) = \infty \big] = 1$.

Define the nonincreasing process $I \dfn \inf_{t \in [0, \cdot]} \hX(t)$; then, $I(\infty) = \inf_{t \in \Real_+} \hX(t)$ is the overall minimum of $\hX$. Let $\bG = (\G(t))_{t \in \zi}$ be the smallest filtration satisfying the usual hypotheses, containing $\bF$, and making $I(\infty)$ a $\G(0)$-measurable random variable. Consider any random time $\rho$ such that $\hX (\rho) = \inf_{t \in \Real_+} \hX(t) = I (\infty)$ --- in other words, $\hX$ achieves at $\rho$ its overall minimum. Since $\prob \big[ \limt \hX(t) = \infty \big] = 1$, such a time is $\prob$-a.s. finite --- in fact, it is also $\prob$-a.s. unique, as will be revealed in Theorem \ref{thm: main} below. Therefore, $\prob$-a.s., $\rho = \inf \big \{ t \in \Real_+ \such \hX (t) = I(\infty) \big \}$, the latter being a stopping time on $\basisg$; since $\G(0)$ contains all $\prob$-null sets of $\F$, it follows that $\rho$ is a stopping time on $\basisg$. Therefore, $\bG$ is strictly larger than the smallest filtration that satisfies the usual hypotheses, contains $\bF$, and makes $\rho$ a stopping time.

What follows is the first result of the paper --- its proof is given in Section \ref{sec: proof}.

\begin{thm} \label{thm: main}
Assume that condition \emph{NA$_1$} holds and that the discounting process is asymptotically suboptimal. Then, the time of minimum of $\hX$ is $\prob$-a.s. unique. With $\rho$ denoting such a time, the process $S^\rho = \pare{S(\rho \wedge t)}_{t \in \Real_+}$ is a local martingale on $\basisgp$.
\end{thm}

\begin{rem}
The result of Theorem \ref{thm: main} does not appear to follow directly from well known results in the theory of filtration enlargements. In particular:
\begin{itemize}
	\item A widely used sufficient condition that enables the use of the theory of initial filtration enlargements is the so-called \emph{Jacod's criterion} \cite{Jacod-init}, which states that the conditional law of the random variable $I(\infty)$ given $\F(t)$ is absolutely continuous with respect to its unconditional law for all $t \in \Real_+$. However, the conditional law of $I(\infty)$ given $\F(t)$ has a Dirac component of mass $1 - I(t) \hY(t)$ at the point $I(t)$, as follows from Doob's maximal identity (\cite[Lemma 2.1]{MR2247846} --- see also the beginning of Section \ref{sec: proof}), while the unconditional law of $I(\infty)$ is standard uniform (this is proved in Section \ref{sec: proof}). Therefore, Jacod's criterion fails.
	\item The Jeulin-Yor semimartingale decomposition result (see \cite{MR519998}) cannot be utilized, because this is not a case of progressive filtration enlargement. Furthermore, as already noted, the filtration $\bG$ is strictly larger than the smallest filtration that satisfies the usual hypotheses, contains $\bF$, and makes $\rho$ a stopping time. 
\end{itemize}
One could use the general results of \cite[Section 3]{MR2247846} in order to establish the validity of Theorem \ref{thm: main}. Here, we provide a simple, self-contained alternative proof, in the course of which the concepts of local martingale deflators and martingale measures will play an important role.
\end{rem}

\begin{rem} \label{rem: generalizing myself}
Theorem \ref{thm: main} justifies the title of the paper. With the insider information flow $\bG$, investing in the risky assets before time $\rho$ gives the same instantaneous return as the locally riskless asset, but entails (locally) higher risk; therefore, before $\rho$ an insider would not be willing to take any position on the risky assets. One can make the point more precise. Let $\X_\bG^\rho$ be the class of nonnegative processes of the form $x + \int_0^\cdot \vartheta^\top(t) \ud S^\rho(t)$, where now $x$ is $\G(0)$-measurable and $\vartheta$ is $\bG$-predictable and $S^\rho$-integrable. By Theorem \ref{thm: main}, all processes in $\X_\bG^\rho$ are nonnegative local martingales on $\basisgp$, which implies that they are nonnegative supermartingales on $\basisgp$. Therefore, $\expec [X (\rho) \such I(\infty)] \leq X(0)$ holds for all $X \in \X_\bG^\rho$. (In particular, $\expec [X (\rho)] \leq X(0)$ holds for all $X \in \X_\bF$, which sharpens the conclusion of \cite[Theorem 2.15]{Kar-pref} for continuous-path semimartingale models.) Jensen's inequality then implies that any expected utility maximizer having an increasing and concave utility function, information flow $\bG$, and time-horizon before $\rho$, would not take any position in the risky assets. 
\end{rem}

\begin{rem}
At first sight, Theorem \ref{thm: main} appears counterintuitive. If the overall minimum of $\hX$ is known from the outset exactly, and especially if it is going to to be extremely low, taking an opposite (short) position in it should ensure particularly good performance at the time of the overall minimum of $\hX$. Of course, admissibility constraints prevent one from taking an \emph{absolute} short position on the \num \ portfolio; still, one can imagine that a \emph{relative} short position on the \num \ portfolio should result in something substantial. To understand better why this intuition fails, remember that $\hX = \Exp \pare{\int_0^\cdot \xi^\top(t) \ud S(t)}$ in the notation of \S \ref{subsec: set-up}, which was noted in the discussion before Definition \ref{dfn: disc process}. A relative short position would result in the wealth $X = \Exp \pare{ - \int_0^\cdot \xi^\top(t) \ud S(t)}$. Straightforward computations show that
\[
X (\rho) = \frac{1}{\hX(\rho)} \exp \pare{- \int_0^\rho \pare{\xi^\top(t) c(t) \xi(t)} \ud G(t)}
\]
Note that the term $\int_0^\rho \pare{\xi^\top(t) c(t) \xi(t)} \ud G(t)$ is the integrated squared volatility of the \num \ portfolio up to time $\rho$, as follows from $\hX = \Exp \pare{\int_0^\cdot \xi^\top(t) \ud S(t)}$. Even though $\hX(\rho)$ can be very close to zero, the term $\exp \pare{- \int_0^\rho \pare{\xi^\top(t) c(t) \xi(t)} \ud G(t)}$ will compensate for the small values of $\hX(\rho)$. In effect, the integrated squared volatility of the \num \ portfolio up to the time of its overall minimum will eliminate any chance of profit by taking short positions in it. 
\end{rem}

\subsection{A partial converse to Theorem \ref{thm: main}}

In Remark \ref{rem: generalizing myself}, it was argued that $\expec [X(\rho)] \leq X(0)$ holds for all $X \in \X_\bF$. A partial converse of the previous result will be presented now. Before stating the result, some  definitions are needed.

\begin{defn} \label{dfn: complete}
Consider a market as described in \S \ref{subsec: set-up}, satisfying condition NA$_1$. The market will be called \textsl{complete} if for any stopping time $\tau$ and any $\F_\tau$-measurable nonnegative random variable $H_\tau$ with $\expec \big[ \hY_\tau H_\tau \big] < \infty$, there exists $X \in \X_\bF$ such that $X_\tau = H_\tau$.
\end{defn}

\begin{rem}
A market as described in \S \ref{subsec: set-up} satisfies condition NA$_1$ \emph{if and only if} there exists at least one strictly positive supermartingale deflator. It can be actually shown that the market is further complete in the sense of Definition \ref{dfn: complete} \emph{if and only if} there exists a unique strictly positive supermartingale deflator. The proof is similar to the one for the case where an equivalent martingale measure exists in the market --- one has to utilize results on optional decomposition under the assumption that a strictly positive local martingale deflator (but not necessarily an equivalent martingale measure) exists in the market; such results are presented in \cite{MR1651229}. In fact, it can be further shown that in a complete market, for any stopping time $\tau$ and $\F_\tau$-measurable nonnegative random variable $H_\tau$, one has
\[
\expec \big[ \hY_\tau H_\tau \big] = \min \set{x \in \Real_+ \such \text{ there exists } X \in \X_\bF(x) \text{ with } X_\tau = H_\tau},
\]
which gives a formula for the minimal hedging price of the payoff $H_\tau$ delivered at time $\tau$.
\end{rem}

\begin{defn}
Let $\phi$ be a random time on $\basisp$. If $\prob \bra{\phi = \tau} = 0$ holds for all stopping times $\tau$ on $\basis$, we shall say that \textsl{$\phi$ avoids all stopping times on $\basisp$}. Furthermore, $\phi$ will be called
an honest time on $\basis$ if for all $t \in \Real_+$ there exists an $\F_t$-measurable random variable $\phi_t$ such that $\phi = \phi_t$ holds on $\set{\phi \leq t}$.
\end{defn}

As it turns out (and will come as an immediate consequence of Theorem \ref{thm: partial conv} below), the random time $\rho$ defined in \S \ref{subsec: main result} is an honest time that avoids all stopping times on $\basisp$. The next result states that, if the market is viable and complete, $\rho$ is the \emph{unique} honest time that avoids all stopping times on $\basisp$, with the property that a wealth processes sampled at this random time has expectation dominated by its initial capital.

\begin{thm} \label{thm: partial conv}
Assume that condition \emph{NA$_1$} holds and that the market is complete. Let $\phi$ be an honest time that avoids all stopping times on $\basisp$, such that $\expec[X(\phi)] \leq X(0)$ holds for all $X \in \X_\bF$. Then, the discounting process is asymptotically suboptimal and $\phi = \rho$.
\end{thm}

\begin{rem}
An inspection of the proof of Theorem \ref{thm: main} shows that, under its assumptions, whenever $\phi$ is the time of maximum of a continuous-path local martingale deflator (which is an honest time that avoids all stopping times), $\expec[X(\phi)] \leq X(0)$ holds for all $X \in \X_\bF$. Therefore, if the market is incomplete, in which case there exist more than one local martingale deflators, the result of Theorem \ref{thm: partial conv} is no longer valid.

Furthermore, since the honest time $\phi = 0$ is such that $\expec[X(\phi)] \leq X(0)$ trivially holds for all $X \in \X_\bF$, the assumption that $\phi$ avoids all stopping times on $\basisp$ cannot be avoided in the statement of Theorem \ref{thm: partial conv}. It is less clear how essential the assumption that $\phi$ is an honest time is. No immediate counterexample comes to mind, although it is quite possible that one exists. Note, however, that $\phi$ being an honest time is instrumental in the proof of Theorem \ref{thm: partial conv}; therefore, further investigation of this issue is not undertaken.
\end{rem}

\section{Proofs} \label{sec: proof}

In the course of the proofs below, we shall use the so-called Doob's maximal identity, which we briefly recall for the reader's convenience. If $M$ is a continuous-path nonnegative local martingale on $\basisp$ such that $\prob[\lim_{t \to \infty} M_t = 0] = 1$ holds, then, with $M^* \dfn \max_{t \in [0,\cdot]} M_t$ and $\rho^M$ denoting \emph{any} time of maximum of $M$, one has the equality $\prob[\rho^M > \tau \such \F_\tau] = M_\tau / M^*_\tau$ whenever $\tau$ is a finite stopping time on $\basis$. Doob's maximal identity can be shown by applying Doob's optional sampling theorem. For a proof of the identity in the form presented above, see  \cite[Lemma 2.1]{MR2247846} combined with \cite[proof of Theorem 2.14]{Kar-pref}.

\subsection{Proof of Theorem \ref{thm: main}}

We first show that $\rho$ is $\prob$-a.s. unique. Define the random times $\rho' \dfn \inf \big\{ t \in \Real_+ \such \hX(t) = I(\infty) \big\}$ and $\rho'' \dfn \sup \big\{ t \in \Real_+ \such \hX(t) = I(\infty) \big\}$. Since $\hY \dfn 1 / \hX$ a nonnegative local martingale that vanishes at infinity on $\basisp$, Doob's maximal identity implies that $\prob \bra{\rho' > t \such \F(t)} = \prob \bra{\rho'' > t \such \F(t)} = I(t) \hY(t)$ for all $t \in \Real_+$. The previous imply that $\rho'$ and $\rho''$ have the same law under $\prob$. Since $\rho' \leq \rho''$, it follows that $\prob [\rho' = \rho''] = 1$. Furthermore, since for any time $\rho$ of minimum  of $\hX$ (which is a time of maximum of $\hY$) we have $\rho' \leq \rho \leq \rho''$, it follows that the time of minimum of $\hX$ is $\prob$-a.s. unique.

For all $u \in \zo$ define $\eta_u \dfn \inf \big\{t \in \Real_+ \such \hY(t) = 1 / (1 - u) \big\}$; then, $(\eta_u)_{u \in \zo}$ is a nondecreasing collection of stopping times on $\basis$. Recall that $\hY$ is a nonnegative local martingale on $\basisp$ such that $\hY(0) = 1$ and $\prob \big[ \lim_{t \to \infty} \hY(t) = 0 \big] = 1$. Also, $1 / I = \sup_{t \in [0,\cdot]} \hY(t)$ and $I(\rho) = I(\infty)$. By the definition of $(\eta_u)_{u \in \zo}$, $\hY^{\eta_u}$ is a uniformly bounded martingale on $\basisp$ with terminal value $\hY^{\eta_u}_\infty = \hY_{\eta_u} = 1 / (1 - u) \indic_{\set{\eta_u < \infty}}$. In particular, Doob's optional sampling theorem gives $\prob \bra{\eta_u < \infty} = 1 - u$; therefore, $I(\infty)$ has the standard uniform distribution under $\prob$ since $\prob \bra{I(\infty) \leq 1 - u} = \prob \bra{\eta_u < \infty} = 1 - u$ holds for $u \in \zo$.

For $u \in \zo$, let $\qprobu$ be the probability $\prob$ on $(\Omega, \, \F)$ conditioned on $\set{\eta_u < \infty}$; of course, $\qprobu$ is absolutely continuous with respect to $\prob$. From the discussion above, $\ud \qprobu / \ud \prob = \pare{1 / (1 - u)} \indic_{\set{\eta_u < \infty}} = \hY (\eta_u) = $ holds for all $u \in \zo$. We use ``$\expecqu$'' to denote expectation under $\qprobu$ for $u \in \zo$ and ``$\expec$'' to denote expectation under $\prob = \prob_0$.

\begin{rem} \label{rem: lmm}
Since all $\hY S^i$, $i \in \set{1, \ldots, d}$, are local martingales on $\basisp$, it follows that $S^{\eta_u}$ is a local martingale in $\basisqu$ for all $u \in \zo$. In other words, $\qprobu$ is an absolutely continuous local martingale measure for $S^{\eta_u}$ for all $u \in \zo$.

A key step towards the proof of Theorem \ref{thm: main} will be Lemma \ref{lem: big help} below. Loosely interpreted, it states that taking the expectation of an $\basis$-optional process sampled at $\rho$ is tantamount to taking the expectation of the same process sampled at $\eta_u$ under $\prob_u$, where ``$u$ has standard uniform distribution, independent of everything else''. Combined with the fact that $\qprobu$ is an absolutely continuous local martingale measure for $S^{\eta_u}$ for all $u \in \zo$, this immediately connects to the statement of Theorem \ref{thm: main}.
\end{rem}

Before stating and proving Lemma \ref{lem: big help}, define also the nonnegative nondecreasing process $U = 1 - I$. Of course, $U(\infty) = U(\rho) = 1 - I(\infty)$ has the standard uniform distribution under $\prob$.

\begin{lem} \label{lem: big help}
For all $u \in \zo$, $\qprobu \bra{\eta_u < \infty} = 1$ holds; in particular, $\qprobu [U(\eta_u) = u] = 1$. Furthermore, for any bounded and $d$-dimensional process $V$ that is optional on $\basis$, we have
\begin{equation} \label{eq: big help}
\expec \bra{V(\rho)} = \expec \bra{ \int_{\Real_+} V(t) \hY(t) \ud U(t) } = \int_{\zo} \expecqu \bra{V(\eta_u)} \ud u.
\end{equation}
\end{lem}

\begin{proof}
First of all, note that $\qprobu[\eta_u < \infty] = \expec [(1/(1-u)) \indic_{\set{\eta_u < \infty}}] = (1/(1-u)) \prob [\eta_u < \infty] = 1$ holds for all $u \in \zo$.

In order to establish \eqref{eq: big help}, start by observing that $\prob [\rho > t \such \F(t)] = I(t) \hY(t)$ holds for all $t \in \Real_+$, in view of Doob's maximal identity. (Recall that $\hY$ is a continuous-path nonnegative local martingale on $\basisp$ and that $\prob [ \lim_{t \to \infty} \hY_t = 0 ] = 1$.) Fix $s \in \Real_+$ and $t \in \Real_+$ with $s \leq t$. The definition of $I$ and the integration-by-parts formula give
\begin{align*}
I(s) \hY(s) - I(t) \hY(t) &= - \int_s^t \hY(v) \ud I(v) - \int_s^t I(v) \ud \hY(v) \\
&= - \int_s^t \frac{1}{I(v)} \ud I(v) - \int_s^t I(v) \ud \hY(v) \\	
&= \log(I(s)) - \log(I(t)) - \int_s^t I(v) \ud \hY(v),
\end{align*}
the second equality following from the fact that $\int_{\Real_+} \indic_{\{\hY(t) \, \neq \, 1/I(t)\}} \ud I(t) = 0$. Note that $0 \leq I \hY \leq 1$. With $(\tau_n)_{\nin}$ denoting a localizing sequence for $\int_0^\cdot I(v) \ud \hY(v)$, which is a local martingale on $\basisp$, it follows that
\begin{align*}
\prob [s \wedge \tau_n < \rho \leq t \wedge \tau_n \such \F(s \wedge \tau_n)] &= \expec \big[ I(s \wedge \tau_n ) \hY(s \wedge \tau_n ) - I(t \wedge \tau_n ) \hY(t \wedge \tau_n ) \such \F(s \wedge \tau_n) \big] \\
&= \expec \bra{ \log(I(s \wedge \tau_n)) - \log(I(t \wedge \tau_n)) \such \F(s \wedge \tau_n)}.
\end{align*}
Upon sending $n$ to infinity, appropriate versions of the bounded and monotone convergence theorem applied to the first and last sides of the above equality will give
\[
\prob [s < \rho \leq t \such \F(s)] = \expec \bra{ \log(I(s)) - \log(I(t)) \such \F(s)}.
\]
As $- \log(I)$ is non-decreasing and adapted, it coincides with the optional compensator (dual optional projection) of $\indic_{\dbraco{\rho, \infty} }$ on $\basisp$. In other words,
\begin{align*}
\expec [V(\rho)] &= \expec \bra{- \int_{\Real_+} V(t) \frac{\ud I(t)}{I(t)}} \\
&= \expec \bra{- \int_{\Real_+} V(t) \hY(t) \ud I(t)} \\
&= \expec \bra{\int_{\Real_+} V(t) \hY(t) \ud U(t)} \\
&= \expec \bra{ \int_{\zo} V (\eta_u) \hY (\eta_u) \indic_{\set{\eta_u < \infty}} \ud u} \\
&= \int_{\zo} \expec \bra{\hY (\eta_u) V(\eta_u) } \ud u \ = \ \int_{\zo} \expecqu \bra{ V (\eta_u) } \ud u,
\end{align*}
the second equality following from the fact that $\int_{\Real_+} \indic_{\{ \hY(t) \, \neq \, 1/I(t) \}} \ud I(t) = 0$ and the fourth by a simple time-change.
The above establishes \eqref{eq: big help} and completes the proof of Lemma \ref{lem: big help}.
\end{proof}

Continuing with the proof of Theorem \ref{thm: main}, we may assume that $S$ is actually bounded via a simple localization argument. In all that follows, fix arbitrary $s \in \zi$ and $t \in \zi$ with $s \leq t$, $B \in \F_{s}$, as well as a bounded deterministic function $f: \zo \mapsto \Real_+$. A use of the $\pi$-$\lambda$ theorem implies that, in order for the result to hold, one only needs to show that $\expec \bra{S^\rho(t) f(U(\infty)) \indic_B} = \expec \bra{S^\rho(s) f(U(\infty)) \indic_B}$. Further noticing that $\prob[U(\infty) = U(\rho)] = 1$, and using the obvious equality $S^\rho(t) f(U(\rho)) \indic_B = S^\rho (s) f(U(\rho)) \indic_B \indic_{\set{\rho \leq s}} + S^\rho(t) f(U(\rho)) \indic_{B} \indic_{\set{\rho > s }}$, one needs to establish
\begin{equation} \label{eq: thing to show}
\expec \bra{S^\rho(t) f(U(\rho)) \indic_B \indic_{\set{\rho > s }}} = \expec \bra{S^\rho(s) f(U(\rho)) \indic_B \indic_{\set{\rho > s }}}
\end{equation}

Since $S$ is assumed bounded, Remark \ref{rem: lmm} implies that $S^{\eta_u}$ is a martingale on $\basisqu$ for all $u \in \zo$. Observe that the process $V \dfn S^t f(U) \indic_{B} \indic_{\dbraoo{s, \infty}}$ is optional on $\basis$; furthermore, $V(\rho) = S^\rho(t) f(U(\rho)) \indic_{B} \indic_{\set{\rho > s}}$. Therefore, from Lemma \ref{lem: big help}, recalling that $\qprobu[U (\eta_u) = u]$ for all $u \in \zo$, we obtain 
\begin{align*}
\expec \bra{S^\rho(t) f(U(\rho)) \indic_{B} \indic_{\set{\rho > s}}} &= \int_{\zo} f(u)  \expecqu \bra{S^{\eta_u} (t) \indic_B \indic_{\set{\eta_u > s}} } \ud u \\
&= \int_{\zo} f(u)  \expecqu \bra{S^{\eta_u} (s) \indic_B \indic_{\set{\eta_u > s}} } \ud u  \ = \ \expec \bra{S^\rho(s) f(U(\rho)) \indic_{B} \indic_{\set{\rho > s}}},
\end{align*}
which is exactly \eqref{eq: thing to show} and completes the proof of Theorem \ref{thm: main}.

\subsection{Proof of Theorem \ref{thm: partial conv}}

To begin with, note that $\basisp$ supports only continuous local martingales. Indeed, otherwise there would exist a nontrivial strictly positive process $N$ with $N(0) = 1$, such that $N$ is a purely discontinuous local martingale on $\basisp$; but then, $N \hY$ would be a strictly positive local martingale deflator in the market, which contradicts the uniqueness of the strictly positive local martingale deflator $\hY$.

Since all local martingales on $\basis$ are continuous and $\phi$ is an honest time that avoids all stopping times on $\basisp$, \cite[Theorem 4.1]{MR2247846} implies that $\phi$ is the time of overall maximum of a nonnegative continuous local martingale $L$ on $\basisp$ with $L(0) = 1$ and $\prob \bra{\limt L(t) = 0} = 1$. We shall show below that $L = \hY$; this shows at the same time that $\phi = \rho$ and that the discounting process is asymptotically suboptimal, the latter following from $\prob \bra{\limt L(t) = 0} = 1$.

As in the proof of Theorem \ref{thm: main}, with $L$ replacing $\hY$ and $\phi$ replacing $\rho$, for all $u \in \zo$ define
$\eta_u \dfn \inf \set{t \in \Real_+ \such L(t) = 1/(1 - u)}$ and
$\qprobu$ via $\ud \qprobu = L (\eta_u) \ud \prob = \pare{1/ (1 - u)}
\indic_{\set{\eta_u < \infty}}$. Define the nondecreasing processes $L^* \dfn \sup_{t \in [0, \cdot]} L(t)$ and $K \dfn 1 - 1/L^*$. Following the reasoning of Lemma \ref{lem: big help} (replacing $\hY$ and $U$ there by $L$ and $K$ respectively --- note that in the proof of Lemma \ref{lem: big help}, we only use the facts that $\hY$ is a nonnegative continuous local martingale on $\basisp$ with $\hY(0) = 1$ and $\prob \big[ \limt \hY(t) = 0 \big] = 1$, properties that $\hY$ shares with $L$), we obtain
\begin{equation} \label{eq: big help 2}
\expec \bra{V(\phi)} = \expec \bra{ \int_{\Real_+} V(t) L(t) \ud K(t)},
\end{equation}
holding for all nonnegative optional process $V$ on $\basis$.

\begin{lem} \label{lem: big help 2}
For uniformly bounded $X \in \X_\bF$, we have
\begin{equation} \label{eq: key 2}
\expecqu \bra{\int_{0}^{\eta_u} (1 - K(t)) \ud X(t)} \leq 0, \text{ for
all } u \in \zo.
\end{equation}
\end{lem}

\begin{proof}
Let $B \dfn \int_{\zd} X(t) \ud K(t)$; clearly, $B$ is a uniformly
bounded nondecreasing continuous and adapted process on $\basis$. Fix $u \in \zo$. Using integration-by-parts, write
\begin{align*}
\int_{\Real_+} X^{\eta_u}(t) L(t) \ud K(t) &= \int_{0}^{\eta_u} L(t) \ud
B (t) + X(\eta_u) \int_{\eta_u}^\infty L(t) \ud K(t)      \\
&= \int_{0}^{\eta_u} L(t) \ud
B (t) + X(\eta_u) \int_{\eta_u}^\infty L(t) \frac{1}{\pare{L^*(t)}^2} \ud L^*(t)      \\
&= \int_{0}^{\eta_u} L(t) \ud
B (t) + X(\eta_u) \int_{\eta_u}^\infty \frac{1}{L^*(t)} \ud L^*(t)      \\
&= L (\eta_u) B (\eta_u) - \int_{0}^{\eta_u} B(t) \ud L(t) + X (\eta_u)
\pare{\log(L^* (\infty)) + \log(1-u)} \indic_{\set{\eta_u < \infty}},
\end{align*}
the third equality following from the fact that $\int_{\Real_+} \indic_{\{ L(t) \, \neq \, L^*(t) \}} \ud L^*(t) = 0$.
Now, observe that $\expec \bra{L(\eta_u) B(\eta_u)} =
\expecqu[B(\eta_u)] = \expecqu[\int_{0}^{\eta_u} X(t) \ud K(t)]$ and
$\expec \bra{\int_{0}^{\eta_u} B(t) \ud L(t)} = 0$, the latter following from the
facts that $B$ is uniformly bounded and $L^{\eta_u}$ is a uniformly
bounded martingale on $\basisp$. Furthermore, using Doob's maximal identity we obtain that
\[
\expec \bra{\log(L^*(\infty)) + \log(1-u) \such \F(\eta_u)} = 1 \text{ holds on } \set{\eta_u < \infty}.
\]
Therefore, $\expec \bra{X(\eta_u) \pare{\log(L^*(\infty)) + \log(1-u)} \indic_{\set{\eta_u < \infty}}} = \expec \bra{X(\eta_u) \indic_{\set{\eta_u < \infty}}} = (1 - u) \expecqu[X(\eta_u)]$. In view of the fact that $\expec \bra{\int_{\Real_+} X^{\eta_u}(t) L(t) \ud K(t)} =
\expec[X^{\eta_u} (\phi)] \leq X (0)$, as follows from \eqref{eq: big help 2} and the assumptions of Theorem \ref{thm: partial conv}, all the previous give
\[
\expecqu \bra{\int_{0}^{\eta_u} X(t) \ud K(t) + (1 - u) X(\eta_u)} \leq X(0).
\]
Since $\int_{0}^{\eta_u} X(t) \ud K(t) =  K(\eta_u) X(\eta_u) - \int_0^{\eta_u} K(t)
\ud X(t) = u X(\eta_u) - \int_0^{\eta_u} K(t)
\ud X(t)$ holds on $\set{\eta_u < \infty}$ and $\qprobu[\eta_u < \infty] = 1$, we furthermore obtain
\[
\expecqu \bra{X(\eta_u) - \int_{0}^{\eta_u} K(t) \ud X(t)} \leq X(0),
\]
which is the same as \eqref{eq: key 2} and proves Lemma \ref{lem: big help 2}.
\end{proof}

Continuing, for each $i \in \set{1, \ldots, d}$ and $\nin$, define $\tau_n^i \dfn \inf \set{t \in \Real_+ \such |S^i(t) - S^i(0)| \geq
n}$, which is a stopping time on $\basis$. Furthermore, define $X^i_n
\dfn 1 + n^{-1} (S^i - S^i(0))^{\tau^i_n}$ --- it is clear
that $X^i_n \in \X_\bF(1)$ and that $0 \leq X^i_n \leq 2$. For an arbitrary stopping time $\tau$ on $\basis$, apply \eqref{eq: key 2} with
$(X^i_n)^{\tau}$ replacing $X$; one then obtains $\expecqu
\bra{\int_{0}^{\eta_u \wedge \tau^i_n \wedge \tau} (1 - K(t)) \ud
S^i(t)} \leq 0$. Performing exactly the previous work by redefining
$X^i_n \dfn 1 - n^{-1} (S^i - S^i(0))^{\tau^i_n}$, one
obtains $\expecqu \bra{\int_{0}^{\eta_u \wedge \tau^i_n \wedge \tau}
(1 - K(t)) \ud S^i(t)} \geq 0$. In other words, $\expecqu
\bra{\int_{0}^{\eta_u \wedge \tau^i_n \wedge \tau} (1 - K(t)) \ud
S^i(t)} = 0$ holds for all $i \in \set{1, \ldots, d}$, $\nin$, and
any stopping time $\tau$ on $\basis$. This implies that each process
$\int_{0}^{\eta_u \wedge \cdot} (1 - K(t)) \ud S^i(t)$ is a local
martingale on $\basisqu$. Since $1 - K > 0$, we further obtain that
each process $(S^i)^{\eta_u}$ is a local martingale on $\basisqu$. By
the definition of the collection $(\qprobu)_{u \in \zo}$, we conclude that $L S^i$ is a local martingale on $\basisp$ for all $i \in \set{1, \ldots, d}$. This implies that $L$ is a local martingale deflator.
Since $1 / \hX$ is the unique local martingale deflator, we finally
conclude that $L = 1/\hX$, which proves Theorem \ref{thm: partial conv}.

\bibliographystyle{siam}
\bibliography{rand_times}
\end{document}